%% file: rate-hiding.tex
\documentclass[10pt, conference]{IEEEtran}
\input{macros}

\title{\huge\scheme: A Crypto-Coded Modulation Scheme for Rate Information Concealing and Robustness Boosting}
\author{\IEEEauthorblockN{Triet D. Vo-Huu and Guevara Noubir}
\IEEEauthorblockA{College of Computer and Information Science\\
Northeastern University\\
Boston, MA, USA\\
\{vohuudtr, noubir\}@ccs.neu.edu}
}
\begin{document}
\maketitle

\input{abstract}

\input{introduction}

\input{challenges}

\input{approach}

\input{GTCM}

\input{ITPSK}

\input{cryptointl}

\input{evaluation}

\input{conclusion}

\bibliographystyle{IEEEtran}
\bibliography{references,gn-1}

\end{document}

%% file: macros.tex
\usepackage[font={small}]{caption}
\usepackage[font={scriptsize}]{subcaption}
\usepackage{times}
\usepackage{flushend}
\usepackage[numbers,sort&compress]{natbib}
\usepackage{array}
\usepackage{footnote}
\usepackage{url}
\usepackage{xspace}

\newcolumntype{R}{>{\centering}X}

\usepackage{clrscode3e}
\usepackage{amsmath,amssymb}
\usepackage{amsthm}
\usepackage{graphicx}
\usepackage{mathtools}
\usepackage{multirow}

\newcommand{\schemelong}[0]{Conceal and Boost Modulation\xspace}
\newcommand{\scheme}[0]{CBM\xspace}

\newcommand{\Hide}[1]{}

\newcommand{\fref}[1]{Fig.~\ref{#1}}
\newcommand{\K}[0]{\mathcal{K}}
\newcommand{\M}[0]{\mathcal{M}}
\newcommand{\N}[0]{\mathcal{N}}

\newtheorem{lemma}{Lemma}

\makeatletter

\def\old@comma{,}
\catcode`\,=13
\def,{%
  \ifmmode%
    \old@comma\discretionary{}{}{}%
  \else%
    \old@comma%
  \fi%
}
\makeatother

%% file: abstract.tex
\begin{abstract}
Exposing the rate information of wireless transmission enables highly efficient attacks that can severely degrade the performance of a network at very low cost. In this paper, we introduce an integrated solution to conceal the rate information of wireless transmissions while simultaneously boosting the resiliency against interference. The proposed solution is based on a generalization of Trellis Coded Modulation combined with Cryptographic Interleaving. We develop algorithms for discovering explicit codes for concealing any modulation in {BPSK, QPSK, 8-PSK, 16-QAM, 64-QAM}. We demonstrate that in most cases this modulation hiding scheme has the side effect of boosting resiliency by up to 8.5dB.
\end{abstract}

%% file: introduction.tex
\section{Introduction}
\label{sec:intro}

Wireless communication is the key enabling technology of the Mobile Revolution that we are currently enjoying. Beyond enabling Mobile Phones, Wireless Sensor Networks, and the Internet of Things devices, it is also key to Cyber-Physical Systems such as SCADA Wireless Remote Terminal Units. In general, the convenience of wireless communication is making it the primary technology for connecting to the Internet. Manufacturers of laptops and streaming devices, such as the Apple MacBook Pro and the Roku streaming player, are even removing Ethernet ports and entirely relying on wireless. However, the broadcast nature of the wireless medium makes it vulnerable to two types of major attacks \textit{denial of service}, and \textit{information leakage}. In this paper, we are interested in the impact of exposing the rate information (modulation and coding) on enabling such attacks. In related work~\cite{NoubirRST2011}, Noubir et al. showed that knowledge of the rate used in a transmission enables selective jamming of packets resulting in very efficient attacks on all the Wi-Fi rate adaptation protocols they investigated. They analytically and experimentally demonstrated that smart rate-aware jamming can degrade and maintain the performance of a communication link at its lowest rate, while simultaneously blocking other links communications. They showed that when a link is degraded from 54 Mbps to 1 Mbps, it blocks other devices but also results in higher collisions provoking a long-lasting network-wide congestion collapse. One of the key reasons why such attacks are possible is because the rate information is either explicit (e.g., PSF field of the IEEE 802.11 PLCP header) or implicit (analysis of I/Q modulation). 
	
Designing countermeasures to wireless DoS attacks before they become widespread is very important for both military and commercial applications as \textit{Jamming} represents a real and serious threat. A highly mediatized incident in March 2012~\cite{cnn}, resulted in the FCC  releasing an urgent customer advisory cautioning against the import and use of jamming devices~\cite{fcc2}. Since then, the FCC has stepped up its education and enforcement effort~\cite{FCC-Education-Enforcement} and rolled out a new jammer tip line (1-855-55NOJAM)~\cite{FCC-tip-line}, and issued several fines~\cite{FCC-fine}.

\begin{figure}
\centering
\includegraphics[scale=0.4]{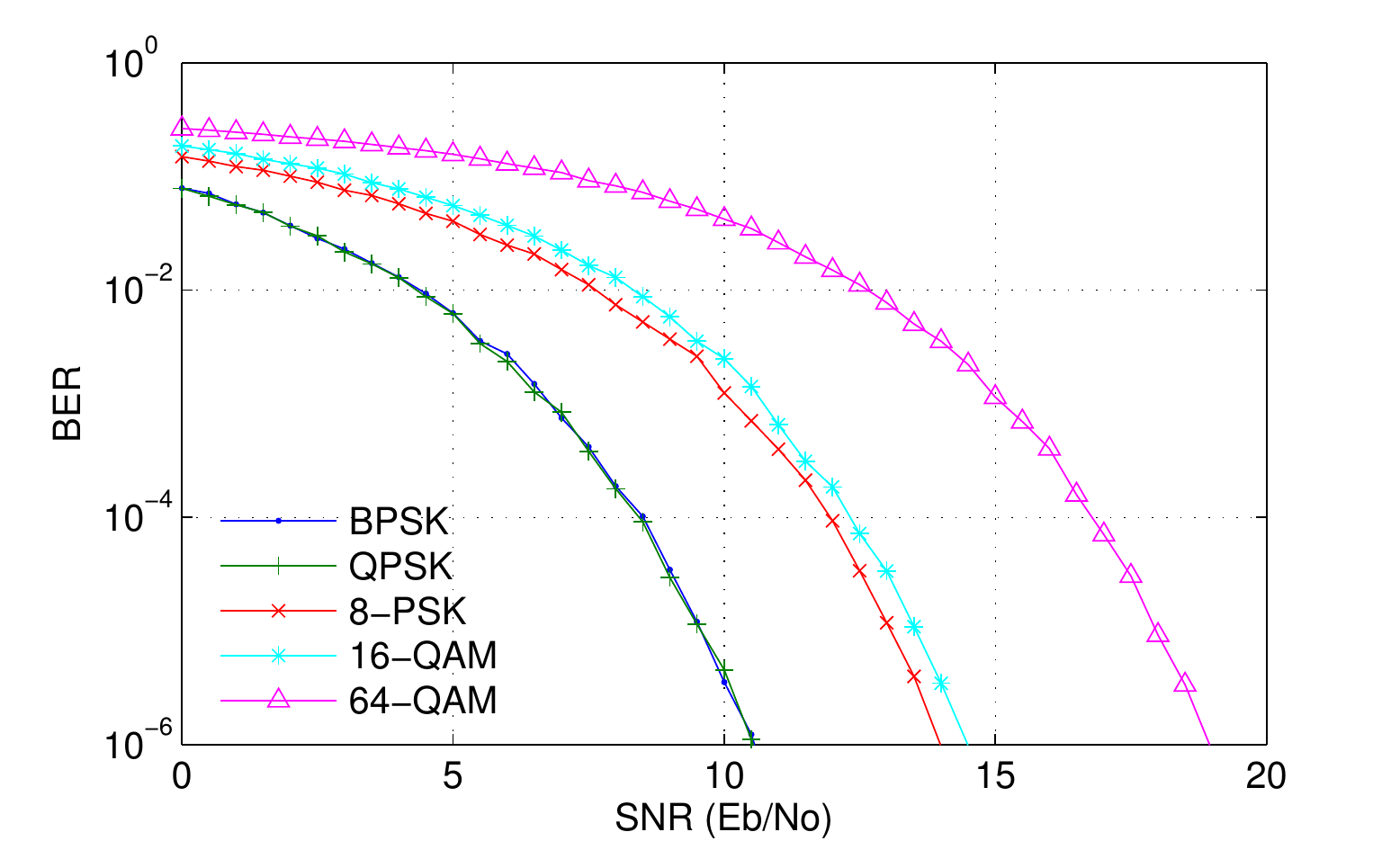}
\caption{Comparison of bit error rates of uncoded modulations.}
\label{fig:uncoded}
\end{figure}

Jamming at the physical layer has been extensively studied (cf.~\cite{poisel2011modern} for a recent monograph on this subject). Reliable communication in the presence of adversaries regained significant interest in the last few years, as new jamming attacks as well as need for more complex applications and deployment environments have emerged (cf.~\cite{jammers,ChenXTZ2009} and references therein). 
Specifically crafted attacks and counter-measures have been studied for packetized wireless data networks~\cite{LinN05,LiKP07,WilhelmMSL11},
multiple access resolution~\cite{BenderFHKL05,BayraktarogluKLNRT2008,antijam-single-hop}, 
multi-hop networks~\cite{XuKWY06,TagueSNP08,LiKP07}, MIMO systems~\cite{KashyapBS2004}, 
and wireless sensor networks~\cite{XuKWY06,XuTZ07,XuTZ08}, spread-spectrum without shared 
secrets~\cite{StrasserCSM08,JinNT09,LiuNDNW2010}.

However, limited work has been done on strengthening rate adaptation algorithms~\cite{GC12,PelechrinisKGB09,FirouzbakhtNS2012}. An unresolved challenge remained to prevent an adversary from guessing the rate (Modulation, Coding information) during the transmission and therefore from selectively interfering with packet. More recently, Rahbari and Krunz proposed a modulation level encryption to hide the rate of communications~\cite{Rahbari:2014:FCM:2627393.2627415}. While this scheme conceals the rate information, it does so at the cost of degrading the robustness of the communication by 1-2dB. 

In this paper, we introduce an integrated solution -- \schemelong{} (\scheme) -- to conceal the rate information of wireless transmissions while simultaneously boosting the communication resiliency against interference. The adversary sees all communications using the highest order modulation. The proposed solution is based on a generalization of Trellis Coded Modulation combined with Cryptographic Interleaving. We developed novel algorithms for efficiently discovering and validating new trellis codes capable of upgrading any modulation constellation to any higher order constellation. We devise explicit codes for concealing any modulation in {BPSK, QPSK, 8-PSK, 16-QAM, 64-QAM}. We demonstrate that in most cases this modulation hiding scheme has the side effect of boosting resiliency by up to 8.5dB and over 9.5 dB in comparison with prior work~\cite{Rahbari:2014:FCM:2627393.2627415}.

A side effect of the proposed mechanisms is that they also help mitigating passive attacks against users traffic analysis and fingerprinting~\cite{AtkinsonARMM13}.

%% file: challenges.tex
\section{Challenges to Concealing Rate Information}
\label{sec:challenges}

We present the model for the wireless system studied in this work, then  discuss
the threats and challenges to the communications due to exposing the rate information.

\subsection{Model}

\paragraph*{Adversary} The adversary is assumed to be located within the wireless communication range
of the legitimate transmitters and receivers.
Besides, the adversary can carry some signal processing on the received signal to fingerprint and extract rate information.
The adversary can also generate power-limited interference to disrupt the communication.

\paragraph*{Transmitter and receiver} The communication between the transmitter and receiver is 
carried hardware/software implementing a public standard, i.e., transmitted and received on some specific channels
at regulated power levels. This assumption implies that the transmission can be detected,
analyzed and jammed by the adversary.

\paragraph*{System robustness} Wireless communications are typically made robust through the use of multiple Modulation-Coding Schemes (MCS) combined with a rate adaptation algorithm that select the rate as a function of link quality. 


\subsection{Challenges}
Knowing the rate being used by the communication can lead to a very efficient attack for the adversary,
because most of wireless systems can only operate reliably at a bit error rate of $10^{-6}$ or below.
For instance, a TCP packet of typical size 1440 bytes can only be transmitted at
a success probability of $99\%$ if the bit error rate of the channel is under $10^{-6}$.
At the bit error rate $10^{-6}$, a communication using 64-QAM modulation requires the transmitter
to transmit at a power $18$dB higher than the noise level (as shown in~\fref{fig:uncoded}).
This implies that an adversary only needs to use a jamming power of about $18$dB ($60$ times) lower than the transmitter's power
to make the communication unreliable. In contrast, a BPSK communication requires a stronger
adversary to achieve the same jamming impact. Smart adversaries
have been built to efficiently jam the wireless link based on the knowledge of the transmission rate
(see~\cite{NoubirRST2011} for an example of adaptive jamming in IEEE 802.11 networks).
Therefore, concealing the rate information is crucial for the robustness of practical wireless systems,
and is the focus of this work. In the following paragraphs, we show that without specific techniques by design, an adversary can easily indentify the rate information and launch efficient attacks. Note that straightforward solutions do not work, e.g., encryption cannot be done before the channel coding, and straightforward modulation upgrade with coding compensation is not efficient (more discussion in Section~\ref{sec:approach}).

\paragraph{Explicit rate information}
In many communication protocols, the rate information of a transmission is unprotected.
For instance in IEEE 802.11 networks, the rate is explicitly specified in the
PLCP Signaling Field (PSF) of the physical layer's frames. Similarly in LTE cellular systems, the transmission rate
is specified in the Modulation and Coding Scheme (MCS) field within the Downlink Control Information (DCI), which is
itself encoded using a publicly known fixed rate $1/3$ convolutional code and
QPSK modulation. An adversary can easily synchronize with the  communication
between two parties, analyze the data frames and extract the rate.
This attack is very practical as demonstrated by~\cite{NoubirRST2011}.

\paragraph{Constellation-based guessing}
Even if the rate information is not explicitly provided within the packet header, the adversary can analyze the received signal in complex I/Q form. After
performing the carrier synchronization, frequency and phase offset correction, the adversary
can trace the received constellation pattern and determine the modulation and coding scheme in use (\fref{fig:constel-attack}).
Since most communication standards specify a limited set of modulations and codes,
the rate information can be efficiently guessed by trial and error.
This method does not require the knowledge of the protocol's frame structure.
A practical rate-aware jammer can easily be built on a software-defined radio (e.g., USRP~\cite{Usrp})
by processing the received samples of the transmitted signal, obtaining the rate, and jamming in real time~\cite{WilhelmMSL11,NoubirRST2011}.
Based on the sophistication of guessing, we classify the method into two subcategories:
\emph{modulation guessing} and \emph{code guessing}, in which the latter requires more complex techniques.

\begin{figure}[tb]
\centering
\begin{subfigure}{0.2\textwidth}
\centering
    \includegraphics[scale=0.6]{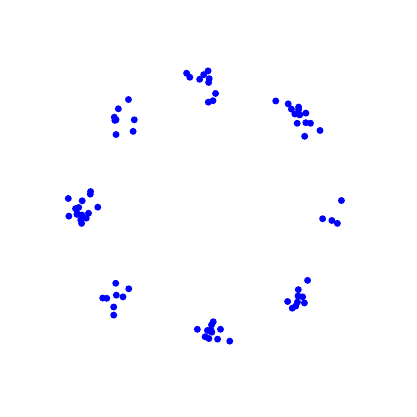}
    \caption{8-PSK}
\end{subfigure}
\quad
\begin{subfigure}{0.2\textwidth}
\centering
    \includegraphics[scale=0.6]{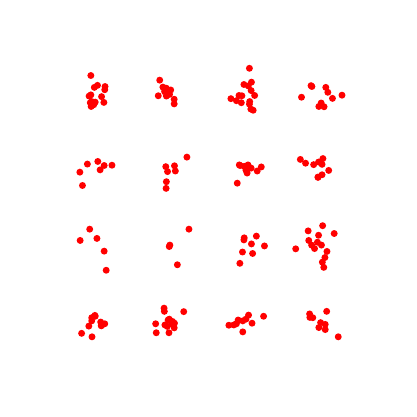}
    \caption{16-QAM}
\end{subfigure}
\caption{Adversary can record the signal and guess the modulation
based on the constellation pattern.}
\label{fig:constel-attack}
\end{figure}

%% file: approach.tex
\section{Approach}
\label{sec:approach}
In this section, we discuss the main ideas and mechanisms of our integrated solution that can not only hide the rate information,
but also increase the robustness of the communication against interference.
Our scheme -- \schemelong (\scheme) -- is depicted in~\fref{fig:approach}.
The General Trellis Coded Modulation (GTCM) module's functionalities are two-fold. First, it makes the constellation pattern indistinguishable to the adversary,
therewith countering the constellation-based modulation-guessing. Second, it boosts the system resiliency against interference.
The Cryptographic Interleaving (CI) module conceals the rate information from explicit exposing and implicit constellation-based code-guessing.

\begin{figure}
\centering
\includegraphics[scale=0.35]{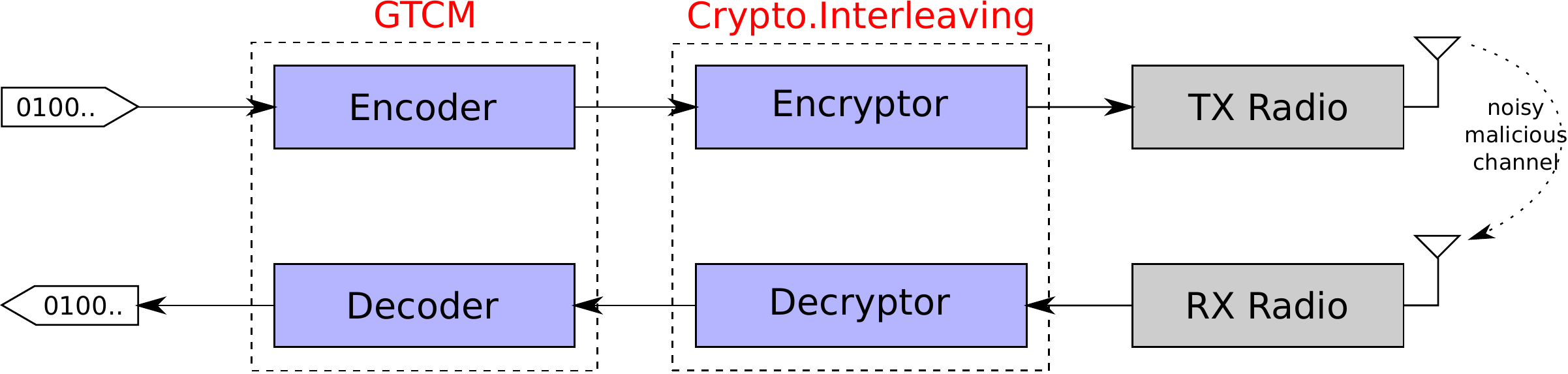}
\caption{Our CBM system comprises General Trellis Coded Modulation and Cryptographic Interleaving modules.}
\label{fig:approach}
\end{figure}

Our idea for hiding the constellation is to
always use a \emph{single unifying} modulation (the highest order) to transmit data in order to create
the same constellation observed by the adversary. To preserve the robustness specified by the
original modulation, the GTCM module encodes the data by a suitable code of rate
matching the bit rate ratio between the original modulation and the target modulation.
To be precise, let's consider a system that
supports a set of $\N$ different modulations ordered by the number of bits per symbol
$b_1\le{}\ldots\le{}b_{\N}$.
Assume that the transmission is carried at a bit rate $k=b_{\K}$ bits per symbol for some modulation $\K$.
In order to conceal the constellation, we transmit using the highest-order modulation $\N$ of bit rate $n=b_{\N}$ as
the target modulation, then encode the data using an adequate code of rate $k/n$
and transmit the encoded data using modulation $\N$.
For example, in a system that supports $\N=4$ different modulations \{BPSK, QPSK, 8-PSK, 16-QAM\},
any transmission using a modulation different than 16-QAM will be encoded and
transmitted with 16-QAM modulation. 
Since the adversary will always observe the same constellation, the GTCM module hides the
actual rate from constellation-based modulation-guessing attacks. Moreover, we argue that
with coding, the robustness of the system can also be improved.

To counter the constellation-based code-guessing, we harden the system using the CI module,
which interleaves the modulated symbols before transmission. We emphasize that the interleaving
process is performed at the baseband samples level, i.e., complex symbols produced by the GTCM module are interleaved 
per block of transmitted symbols. We note that straightforward encryption of data
before the GTCM processing does not conceal the rate information.
For example, the adversary can try to decode the data iterating over all the possible codes.
During the decoding phase, the likelihood of decoding each possible sequence
is recorded and evaluated. The code corresponding to the maximum likelihood
will be the one used by the transmission.
The attack is based on the fact that output sequences of different codes are
not identically distributed in the output stream.
This requires that we design an interleaving mechanism based on cryptographic functions in order to
make the interleaved symbol sequence indistinguishable to the adversary.
We derive a specific method to efficiently generate interleaving functions
used for permuting the output symbols from the GTCM module in such a way that
the transmit stream does not leak the rate information.
For the receiver to be able to decode the data, the rate information is embedded into
the packet in an encrypted form such that only the receiver, who shares the secret key
with the transmitter, can decrypt the information.

It is important to understand the implications of rate hiding.
On one hand, the highest-order modulation creates redundancy by the constellation expansion.
On the other hand, the constellation points' pair-wise distances are closer than in the original constellation.
Without good design specifically targeting to the upgraded modulation's constellation, the system can become
less resilient against interference.
For example, the modulation unification technique used in~\cite{Rahbari:2014:FCM:2627393.2627415}
results in the system robustness 1-2dB less than regular rate-exposing systems.
The reduced resiliency is because no coding is used in their system.
However, even using good traditional binary codes can  not guarantee
the robustness of the system because they maximize the Hamming distance between codewords and are not designed for coded modulation. An illustration is shown in \fref{fig:qpsk-qam16-v6}. We take the best code
$\scriptsize\left(\begin{array}{cccc}
17 & 13 & 05 & 02 \\
10 & 03 & 17 & 15
\end{array}\right)$ of rate $2/4$ from Table VII in~\cite{conv-codes}, and use it with Gray coded 16-QAM modulation.
Comparing it with our derived TCM code
$\scriptsize\left(\begin{array}{cccc}
01 & 12 & 16 & 11 \\
01 & 13 & 16 & 11
\end{array}\right)$ of the same rate and constraint, we see a gain of about 4dB is achieved with our code,
while the binary code almost gives no advantage over uncoded QPSK at $\textrm{BER}=10^{-6}$.
Therefore, with good codes designed for the target modulation, we can gain instead of losing.
As seen later in Section~\ref{sec:gtcm}, increasing the code's constraint length can boost the
resiliency up to 8.5dB.

\begin{figure}
\centering
\includegraphics[scale=0.4]{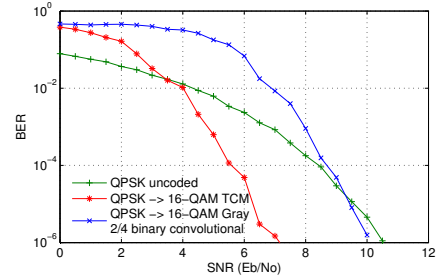}
\caption{Performance comparison between
(1) our TCM code with standard 16-QAM, and
(2) best traditional binary code (from~\cite{conv-codes}) of rate $2/4$ with Gray coded 16-QAM.}
\label{fig:qpsk-qam16-v6}
\end{figure}

Searching for good codes for the rate-hiding systems must take
into account the constellation description defined by the highest-order modulation.
This idea is rooted in the Trellis Coded Modulation (TCM) technique introduced by Ungerboeck~\cite{tcm}, who focused on devising modulation codes of rate $k/(k+1)$. 
In the literature, finding good TCM codes is a challenging problem, for which
only heuristic solutions have been studied such as the set partitioning rules established
in~\cite{tcm}. Unfortunately, a full theoretical analysis is not
yet known for good TCM code construction. In this work, we introduce a new
heuristic approach for upgrading arbitrary modulations. Our heuristic solution is not based on
the conventional set partitioning concept. Instead, we generate the code based on
the general structure of a convolutional code. Specifically, we randomize the mapping between the
inputs, shift registers, and the outputs. Our experimental results show that
this approach can find codes at least as good as the ones found in~\cite{tcm}. 
Yet in some cases, we obtain better codes than those in~\cite{tcm}.
Our work is summarized in the following main points:
\begin{itemize}
\item We develop a set of new algorithms that efficiently discovers good modulation codes of general rate $k/n$. Such codes are not restricted to be uniform as in previous work. Such algorithms include a new technique to efficiently compute the free distance of non-uniform codes.
\item We constructed efficient codes of constraint length up to 10 used for any pair of
modulations in \{BPSK, QPSK, 8-PSK, 16-QAM, 64-QAM\}. For the case of the rate ~$k/(k+1)$ considered by traditional TCM, we found codes that perform better than the ones introduced in~\cite{tcm}.
\item We design a generating method that can
 efficiently produce fast interleaving functions used for concealing the rate information.
 
 \item We evaluate the performance of each constructed code demonstrating a robustness boosting up to 8.5dB and an improvement over related rate hiding techniques of up to 10dB.
\end{itemize}


%% file: GTCM.tex
\section{General Trellis Coded Modulation}
\label{sec:gtcm}
In this section, we decribe the searching procedure for TCM codes
of general rate $k/n$, which are used to encode data
originally modulated by a modulation $\K$ of order $2^k$ to the highest-order
modulation $\N$ of order $2^n$. For convenience, we give a brief overview
on TCM codes.

A TCM code is a convolutional code $(k,n)$ defined by a set of $k$ shift registers
storing the code's $k$ input bits, and a generator matrix which specifies the input-output mapping (e.g.,~\fref{fig:code-example}).
The longer the shift registers are, the more redundancy the codes can have. Thus, we classify
the codes by their constraint length $v$ defined as
$v=\sum_{i=1}^k{}v_i$,
where $v_i$ is the length of the $i$-th shift register.
In our search procedure, we only consider the feed-forward construction of convolutional codes,
because any construction with feedback can be transformed into
a feedback-free construction that produces equivalent codewords~\cite{ecc}.
To represent a code, we use the conventional generator
polynomial form $G(D)=\{g_{ij}(D),i=0\ldots{}k-1,j=0\ldots{}n-1\}$,
where $g_{ij}(D)=\sum_{l=0}^{v_i}a_{l}D^l$ is a univariate polynomial,
and the indeterminate $D$ represents the delay of the input bit in the corresponding shift register.
If $a_l=1$, the $i$-th input's
current value (for $l=0$) and past values (for $l>0$) are mod-2 added (exclusive-or) to the $j$-th output.
For example, the convolutional code $(2,4)$ in~\fref{fig:code-example} has the
generator polynomial
$G=\left[\begin{array}{cccc}D+D^2 & D^3 & D+D^3 & 1+D+D^2 \\
0 & 1 & 0 & 1+D
\end{array}\right]$.

Different than binary convolutional codes whose performance depends on the Hamming distance
of the binary output symbols, TCM codes' performance is determined by the
free \textit{Euclidean} distance $d^{\infty}$, which is the minimum Euclidean distance of any
two complex symbol sequences produced by the code and modulation $\N$.
Since binary codes are not designed for coded modulation, they do not take into account 
the constellation mapping. The best binary code with optimized Hamming distance can have
arbitrarily small Euclidean distance between transmitted complex symbols and result
in low performance when combined with a specific modulation (e.g., \fref{fig:qpsk-qam16-v6}).
In the search for good TCM codes, we use as comparison metrics
the asymptotic coding gain ratio measured by $\beta=d_{\N}^\infty/\Delta_{\K}$, where
$\Delta_{\K}$ is the minimum Euclidean distance between constellation points in the original modulation $\K$.
Good TCM codes must have high $\beta$ ratio.

\begin{figure}
\centering
\includegraphics[scale=0.3]{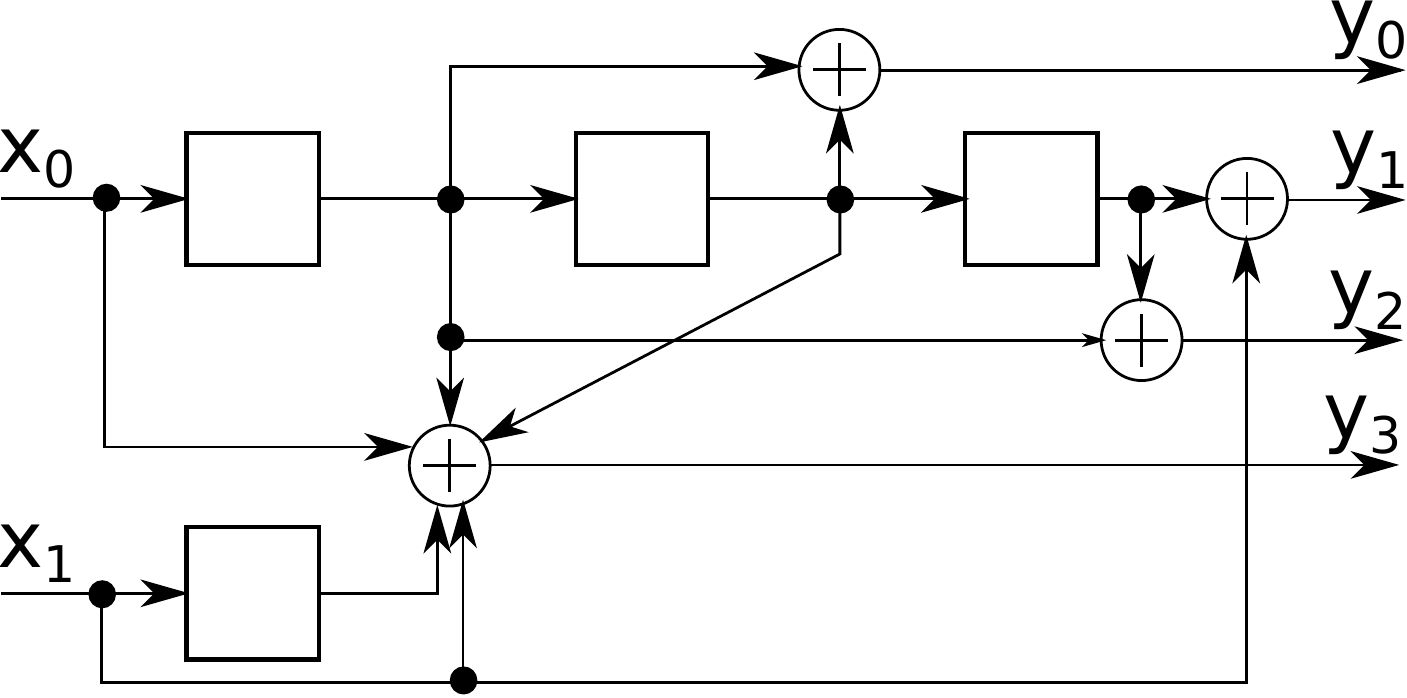}
\caption[example]{Our best TCM code of rate $2/4$ (QPSK $\rightarrow$ 16-QAM) and constraint length $4$.
Its boosting gain over uncoded QPSK is $3.8$dB.} 
\label{fig:code-example}
\end{figure}

\subsection{Code search algorithm}
We introduce a new heuristic approach for searching for good TCM codes. The idea of the algorithm
is that for a given code specification $(k,n,\{v_i\})$, the coefficients of the
generator polynomials $g_{ij}$ are randomized between values $0$ and $1$.
Since each coefficient change corresponds to a new code construction,
we check the generated code for the free Euclidean distance.
The search is performed for a fixed number of trials independent of
the code specification, thus it is extremely faster than a full search which
evaluates all possible codes. Nevertheless, as shown in Section~\ref{sec:search-results},
our randomization approach can achieve the same results as a full search.
The search procedure is illustrated in the $\proc{RandomCodeSearch}$ algorithm below.

\begin{codebox}
\Procname{$\proc{RandomCodeSearch}(k,n,\{v_i\},\M,T)$}
\li $d^{\infty}=0$ \>\>\>\> \Comment free distance of current best code
\li \For $i \gets 1$ \To $T$
\li \Do
        $C \gets \func{generateCode}(k,n,\{v_i\})$
\li     \If $\func{valid}(C)$ \>\>\> \Comment non-catastrophic and equiprobable
\li     \Then $d \gets \proc{ComputeDistance}(C,\M,d^{\infty})$
\li         \If $d>d^{\infty}$
\li         \Then
                $d^{\infty} \gets d$  \>\>\>  \Comment update free distance
\li             $C^* \gets C$        \>\>\> \Comment store new best code
            \End
        \End
    \End
\li \Return $(C^*,d^{\infty})$
\end{codebox}

Our random search is characterized by the number of trials $T$
performed by \proc{RandomCodeSearch}. 
Since the code is randomly generated, it might be
(1) \emph{catastrophic}, i.e., there exists a non-zero input sequence that can
produce all-zero output sequence; and (2) \emph{non-equiprobable}, i.e., the output values
are not uniformly distributed, which can help the adversary distinguish the mapping we aim to conceal.
Therefore, the generated code is first validated to be non-catastrophic and equiprobable, then its free distance $d^{\infty}$ is computed.


\subsection{Free distance computing algorithm}
The computational bottlneck of the code search is computing the Euclidean free distance,
since it is performed for every generated code.
In the conventional TCM code construction method based on set partitioning rules~\cite{tcm},
computing the free distance only involves finding the minimum distance to the all-zero sequence.
However, our random search does not apply the set partitioning concept as
we aim to search in a higher-dimension space so that better codes can be found (including non-uniform ones).
As a result, computing free Euclidean distance has to involve
all pairs of output sequences. Nevertheless, as shown later, we devise an efficient algorithm
-- \proc{ComputeDistance} --
whose running time is on average less than $2$ms on a 3GHz CPU computer for the modulations and depths we consider.

The main idea of our $\proc{ComputeDistance}$'s algorithm is based on
traversing the trellis of the code and appropriately updating the
\emph{state-distances}, which we define shortly below.

First, we introduce some convenient notations. Let $I=\{0,\ldots,2^{C.k}-1\}$ be the
set of inputs, $O=\{0,\ldots,2^{C.n}-1\}$ the set of output symbols, and
$\Lambda=\{0,\ldots,2^{C.v}-1\}$ the set of possible states corresponding
to a code $C$. A path $P$ of length $L$ is defined as a sequence of $3$-tuples
$P=\{(S_i,x_i,y_i),i=0\ldots{}L-1\}$, where $S_i\in\Lambda,x_i\in{}I$ are respectively
the state and input of the code at time $i$, and $y_i\in{}O$ is the output symbol due to $S_i$ and $x_i$.
The distance between two paths $P$ and $\tilde{P}$ of length $L$ is computed by
$\func{dist}(P,\tilde{P})=\sum_{i=0}^{L-1}\attrib{\M}{ed}(\attrib{P}{y},\attrib{\tilde{P}}{y})$,
where $\attrib{\M}{ed}(a,b)$ gives the Euclidean distance between two points $a$ and $b$
on the target coded modulation $\M$'s constellation.
Now we define the state-distance $D[S,\tilde{S}]\stackrel{\Delta}{=}\min\{\func{dist}(P,\tilde{P})\}$ of two states $S$ and $\tilde{S}$
as the minimum Euclidean distance between all possible paths of the same
length that end at state $S$ and $\tilde{S}$, respectively.

The main idea of the algorithm is that we update $D[S,\tilde{S}]$
gradually when traversing the trellis with increasing $L$. When two paths $P$ and $\tilde{P}$
merge, i.e., $\attrib{P}{S_{L-1}}=S=\tilde{S}=\attrib{\tilde{P}}{S_{L-1}}$,
the free distance $d^{\infty}$ is checked and updated with $D[S,\tilde{S}]$.

\begin{codebox}
\Procname{$\proc{ComputeDistance}(C,\M,d^{\infty}_{best})$}
\li $D[S,\tilde{S}] \gets \infty$ for all $(S,\tilde{S})\in{}V^2$ \>\>\>\>\>\>\>\hspace{1ex}   \Comment state-distances    \label{li:start-init}
\li $d^{\infty} \gets \infty$    \>\>\>\>\>\>\>\hspace{1ex} \Comment $C$'s free distance
\li \For \kw{each} $S\in{}\Lambda$, $(x,\tilde{x})\in{}I^2$, $x\neq{}\tilde{x}$                  \label{li:diverse}
\li \Do
        $\proc{UpdateDistance}(S,x,S,\tilde{x})$                                                 \label{li:update1}
    \End                                                                                           \label{li:end-init}
\li \Repeat                                                                                        \label{li:start-loop}
\li     \For \kw{each} $(S,\tilde{S})\in{}\Lambda^2,S\neq{}\tilde{S},D[S,\tilde{S}] < d^{\infty}$    \label{li:for}
\li     \Do
	            \For \kw{each} $(x,\tilde{x})\in{}I^2$                                            \label{li:newsegment}
\li             \Do
                    $\proc{UpdateDistance}(S,x,\tilde{S},\tilde{x})$
\li                             \If $d^{\infty} \le{} d^{\infty}_{best}$                            \label{li:checkbest}
\li                             \Then
                                    \Return $d^{\infty}$
                                \End
	            \End
	    \End
\li \Until $(S,\tilde{S})$ not found in line~\ref{li:for}                                          \label{li:end-loop}
\li \Return $d^{\infty}$
\end{codebox}

The algorithm $\proc{ComputeDistance}$ starts by initializing the 
state-distances to the distance between any path $P$ and $\tilde{P}$
starting from any \emph{same} state $S$ (line \ref{li:start-init}--\ref{li:end-init}).
We make the paths diverge from the same state (line~\ref{li:diverse}),
then compute the distance between them (line~\ref{li:update1}).
In the main loop (line \ref{li:start-loop}--\ref{li:end-loop}),
the state-distances are repeatedly updated for each new segment added (line~\ref{li:newsegment})
to the paths until there exist no more state pairs $(S,\tilde{S})$
whose state-distance $D[S,\tilde{S}]$ is less than $d^{\infty}$ (line~\ref{li:end-loop}).
The maintenance and update of state-distances in both the initialization and the main loop
are performed by
$\proc{UpdateDistance}$, which keeps records of $D[S,\tilde{S}]$
for all $S,\tilde{S}$. Whenever two paths $P$ and $\tilde{P}$ merge at a state $S$,
the corresponding state-distance $D[S,S]$ is checked to update $d^{\infty}$.

\begin{codebox}
\Procname{$\proc{UpdateDistance}(S,x,\tilde{S},\tilde{x})$}
\li     $T \gets \attrib{C}{nextState}(S,x)$; $y \gets \attrib{C}{output}(S,x)$
\li     $\tilde{T} \gets \attrib{C}{nextState}(\tilde{S},\tilde{x})$; $\tilde{y} \gets \attrib{C}{output}(\tilde{S},\tilde{x})$
\li     \If $S=\tilde{S}$
\li     \Then $d \gets \attrib{\M}{ed}(y,\tilde{y})$   \>\>\>\>\>\> \Comment for initialization
\li     \Else
\li         $d \gets \attrib{\M}{ed}(y,\tilde{y}) + D[S,\tilde{S}]$   \>\>\>\>\>\> \Comment for main loop
\li     \End
\li     \If $d < D[T,\tilde{T}]$
\li     \Then
	        $D[T,\tilde{T}] \gets d$    															\label{li:update-state-distance}
\li         \If $d < d^{\infty}$ \kw{and} $T=\tilde{T}$ \>\>\>\>\>\> \Comment two paths merge
\li         \Then $d^{\infty} \gets d$
            \End
        \End
\end{codebox}

Recall the random search procedure discussed previously where each generated code is
computed for the free distance, we speed up the search by storing the best free distance $d^{\infty}_{best}$
associated to the best code $C^*$ discovered so far in order to quickly eliminate
the code of free distance shorter than $d^{\infty}_{best}$
(line~\ref{li:checkbest} in \proc{ComputeDistance}).

\begin{figure}
\centering
\includegraphics[scale=0.12]{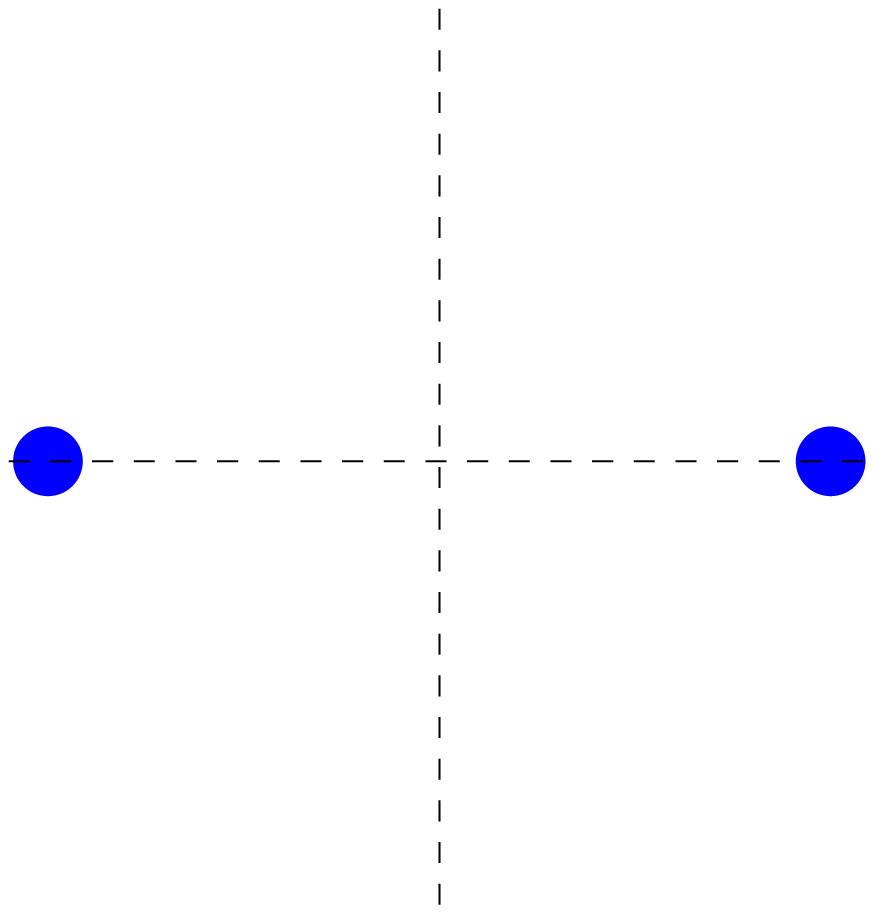}
\includegraphics[scale=0.12]{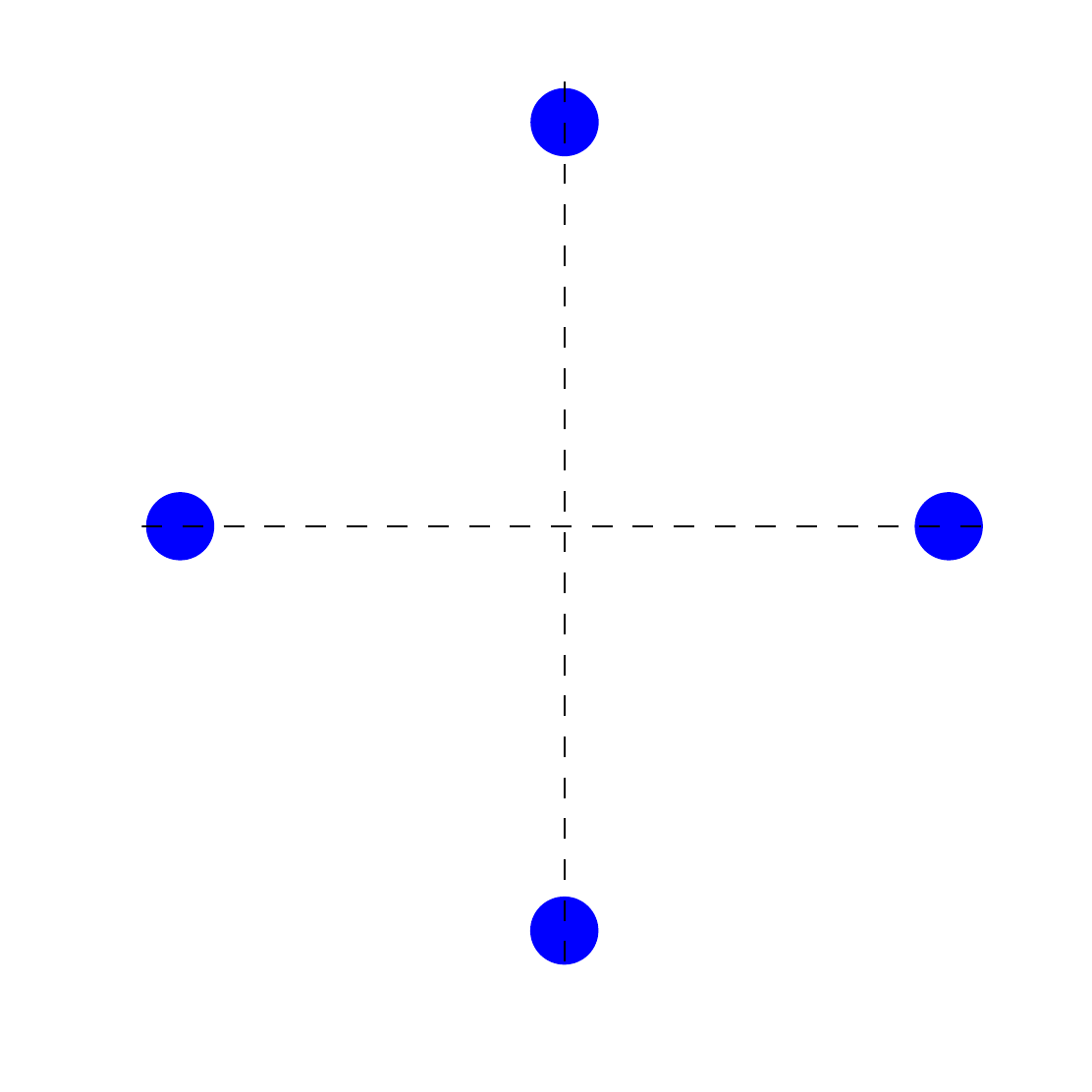}
\includegraphics[scale=0.12]{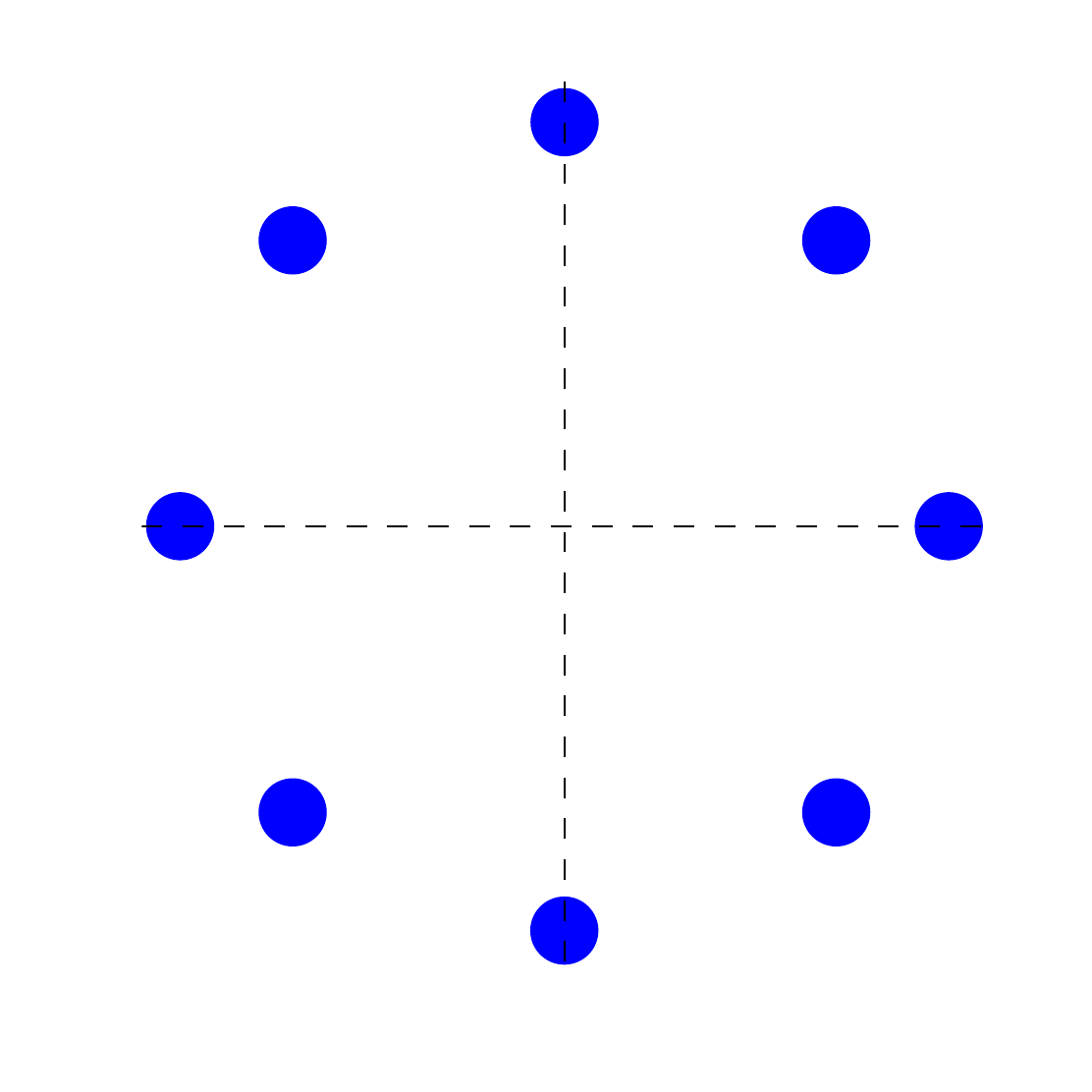}
\includegraphics[scale=0.12]{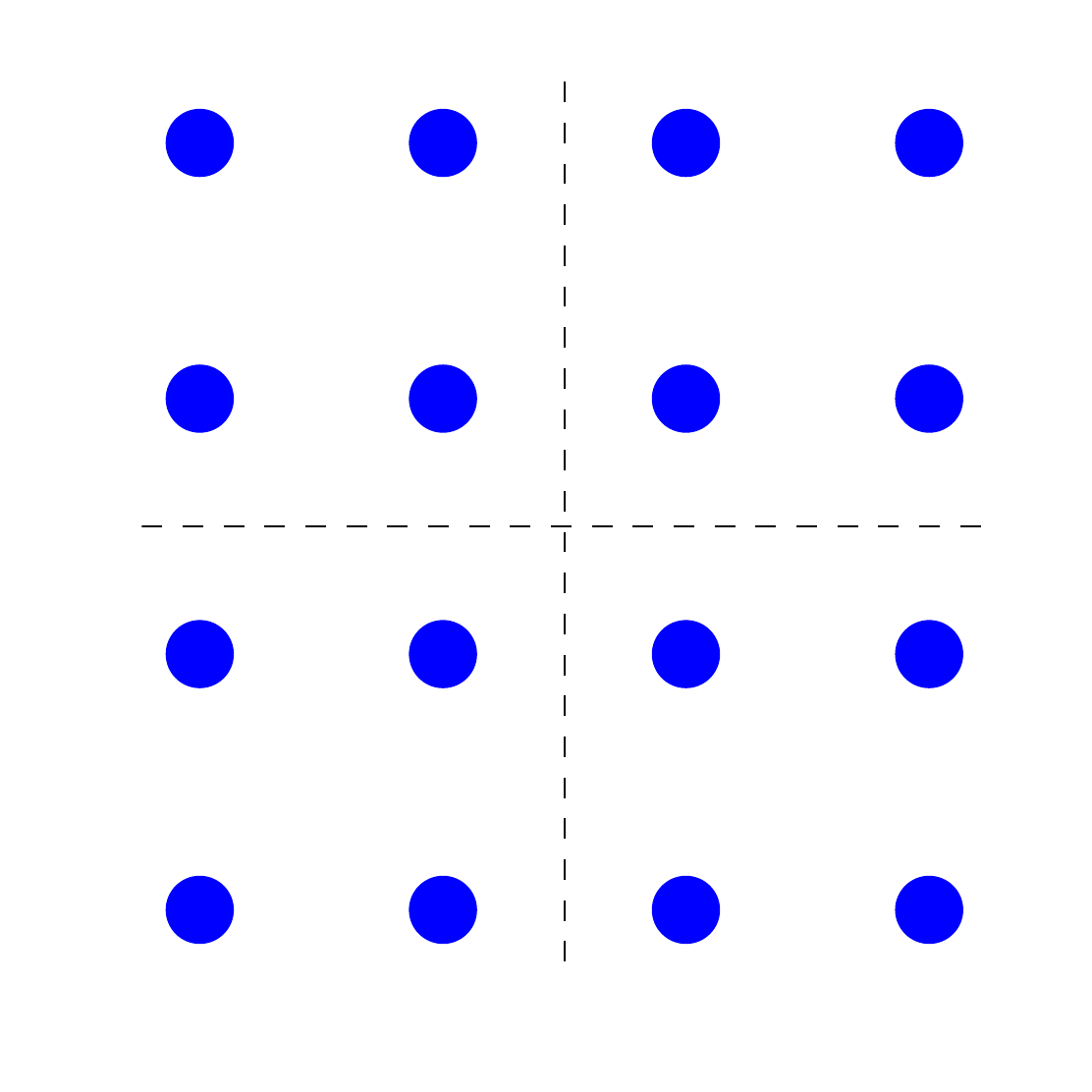}
\includegraphics[scale=0.12]{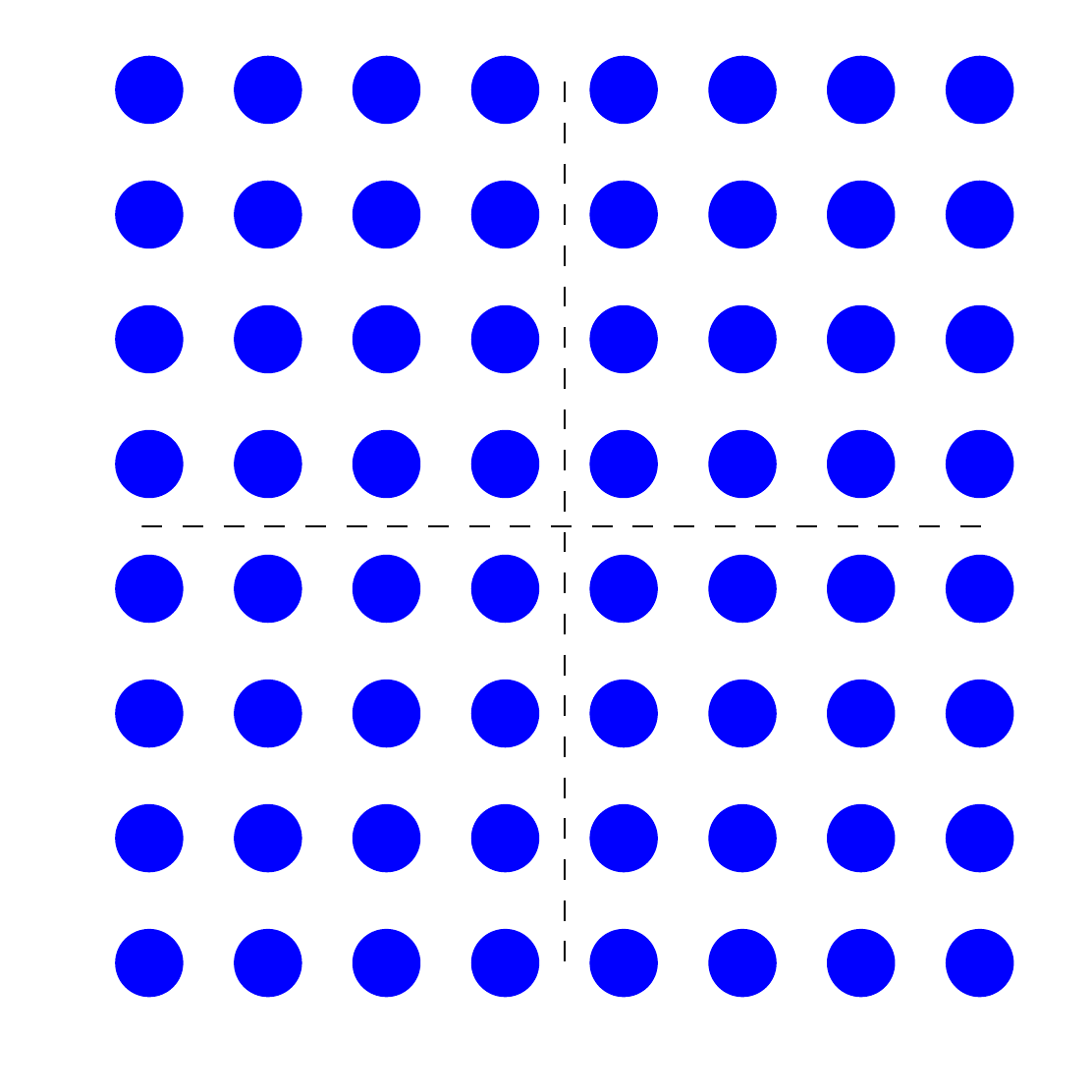}
\caption{BPSK, QPSK, 8-PSK, 16-QAM, 64-QAM constellations.}
\label{fig:constellations}
\end{figure}

\vskip 1ex\noindent{\bf Correctness}
To prove the correctness of the algorithm, we show that 
the state-distances $D[S,\tilde{S}]$ keep records of the distances of all possible
``close'' paths. The proof is based on the following lemma.

\begin{lemma}
\label{lemma1}
At any time $i$ on the code trellis, for any pair of paths $P$ and $\tilde{P}$,
if there exists another pair of paths $Q$ and $\tilde{Q}$ such that 
$P.S_i=Q.S_i$, $\tilde{P}.S_i=\tilde{Q}.S_i$, and $D[Q,\tilde{Q}]<D[P,\tilde{P}]$,
then $P$ and $\tilde{P}$ can be eliminated.
\end{lemma}
\begin{proof}[Proof (sketch)]
By the lemma's assumption, $P$ merges with $Q$ and $\tilde{P}$ merges with $\tilde{Q}$ at time $i$.
It is followed that at time $i+1$, any new pair evolved from $P$ and $\tilde{P}$ will find a similar new pair
evolved from $Q$ and $\tilde{Q}$. Therefore, the pair $(P,\tilde{P})$ cannot have
shortest distance.
\end{proof}

\noindent\emph{Proof of correctness.} At initialization of \proc{ComputeDistance}, $D[S,\tilde{S}]$ are set to non-infinity values
only for pairs of paths starting from the same state (line~\ref{li:diverse}).
This means that $D[S,\tilde{S}]$ properly reflect the distances of paths at initial states.
In the main loop, the algorithm traverses every transition of the trellis
and updates the state-distances.
By Lemma~\ref{lemma1}, the macro $\proc{UpdateDistance}$ will discard
paths corresponding to greater distance $D[S,\tilde{S}]$ and keep
the ones corresponding to the shortest distance so far (line~\ref{li:update-state-distance}).
Therefore, no closest pairs are eliminated by the algorithm.

To see that the algorithm terminates, we show that there exists a time such that 
$D[S,\tilde{S}]\ge{}d^{\infty}$ for all state pairs. It is enough to show
that $D[S,\tilde{S}]$ are increasing while $d^{\infty}$ is decreasing. The former is
correct because evolving paths always contain transitions that results in
positive increment in distance. The latter is due to the update in \proc{UpdateDistance}.
This concludes the proof.

\vskip 1ex\noindent{\bf Computational Complexity}
The time complexity $g(t)$ of \proc{ComputeDistance} depends on the length $L$
of the paths where the free distance is found. We estimate $g(t)$ in the worst case as follows.
First, the initialization of $D[S,\tilde{S}]$ requires $2^{C.v+2C.k}$ calls to \proc{UpdateDistance}.
At each iteration of time $i$ in the main loop, the number of updates is at most
$2^{2C.v+2C.k}$. Therefore, the worst-case complexity of \proc{ComputeDistance} is
$g(t)=2^{C.v+2C.k}+2^{2C.v+2C.k}L=O(2^{2(C.v+C.k)}L)$.
In our experimental search results, we observe that the value of $L$ can be bounded by $L\le{}3v$ for any code.
The running time of the algorithm on a 3GHz CPU computer is less than 2ms.



\subsection{Search results}
\label{sec:search-results}
In this section, we list the good codes found by our randomization approach (Table~\ref{table:bpsk-qpsk-psk8}--\ref{table:qpsk-psk8-qam16-qam64}).
The lists are compiled for codes of constraint length up to $v=10$ and for each pair of original modulation $\K$ and target coded modulation $\N$
in \{BPSK, QPSK, 8-PSK, 16-QAM, 64-QAM\} whose constellations are shown
in~\fref{fig:constellations}. We note that previous work~\cite{tcm} discovered codes for only 1 bit constellation expansion. The asymptotic coding gain $\beta$ is measured in dB, and the generator matrix $G$
is presented in the octal form adopted from~\cite{ecc,conv-codes}. The symbol mapping of m-PSK constellations is
$p(s)=e^{j2\pi{}s/m}$, where $s$ is transmitted symbol, $j=\sqrt{-1}$.
For square m-QAM constellations, the mapping of $s$ is $p(s)=(x_0+s_{L}\Delta_m)+j(y_0+s_{H}\Delta_m)$,
where $\Delta_m$ is the minimum separation in m-QAM, $(x_0,y_0)$ are coordinates
of the zero symbol ($s=0$) located at the bottom-left corner of the constellation, and
$s_{L}=s\bmod{}(\sqrt{m})$, $s_{H}=(s-s_L)/\sqrt{m}$ are corresponding to
low-order and high-order bits of $s$.
Note that for some cases of short constraint length, there are no
good codes with positive boosting gain.

\input{extracted-codes/bpsk-to-qpsk-psk8}

\input{extracted-codes/bpsk-to-qam16-qam64}

\input{extracted-codes/psk8-to-qam16-qam64}

\input{extracted-codes/qam16-qam64}

\input{extracted-codes/qpsk-to-psk8-qam16-qam64-wide}

\vskip 1ex\noindent{\bf Comparison with uniform codes search}
As seen from Table~\ref{table:qpsk-psk8-qam16-qam64} for QPSK $\rightarrow$ 8-PSK, we achieve
better codes than ones in~\cite{tcm} for constraint lengths $v=6,8,10$,
which confirms the intuition that better codes can be discovered
if the uniform mapping property given by the set partitioning rules is relaxed.

\vskip 1ex\noindent{\bf Comparison with full search}
To verify that the good codes discovered using randomization are actually the best for each category,
we perform a full search for some ``small'' tuples\footnote{
Larger values of $\K,\N$ or $v$ make full search run forever in our tests.} $(\K,\N,v)$
and compare the results with codes found by random search.
As we conjectured, the full search does not find any code better than the random search in those cases.
Yet, the random search is extremely fast (cf. Table~\ref{table:search-compare}).
The search results show that good codes are distributed randomly in the search space,
thus searching with randomization is very efficient,
and especially useful for large constraint length and high-order modulations.

\vskip 1ex\noindent{\bf Asymptotic coding gain}
The search results show that with a large enough constraint length there are codes such that
in addition to modulation hiding, the resiliency of the system can be boosted up to $8.5$dB
over uncoded systems, resulting an improvement of up to $10$dB compared to recent work~\cite{Rahbari:2014:FCM:2627393.2627415}.

\input{extracted-codes/searchtime.tex}

%% file: extracted-codes/bpsk-to-qpsk-psk8.tex
\begin{table}
\centering
{\scriptsize
\begin{tabular}{|c|c|c||c|c|c|}
\hline
\multicolumn{3}{|c||}{BPSK $\rightarrow$ QPSK} & \multicolumn{3}{c|}{BPSK $\rightarrow$ 8-PSK} \\
\hline
$v$  &   $\beta$  &        $G$       &    $v$  &   $\beta$  &        $G$     \\
\hline
1    &    1.76    &        (1 3)     &    2    &    3.72    &        (4 2 5) \\
2    &    3.98    &        (2 7)     &    3    &    4.77    &        (10 2 17) \\
3    &    4.77    &        (14 7)     &    4    &    5.36    &        (2 4 31) \\
4    &    5.44    &        (4 35)     &    5    &    6.02    &        (10 2 71) \\
5    &    6.02    &        (33 64)     &    6    &    6.53    &        (12 20 107) \\
6    &    6.99    &        (56 135)     &    7    &    6.99    &        (4 102 251) \\
7    &    6.99    &        (374 147)     &    8    &    7.24    &        (30 102 661) \\
8    &    7.78    &        (232 561)     &    9    &    7.63    &        (400 336 1715) \\
9    &    7.78    &        (665 1312)     &    10    &    7.78    &        (14 1400 3575) \\
10    &    8.45    &        (1256 2175)     &        &        &    \\
\hline
\end{tabular}
}
\caption{BPSK $\rightarrow$ QPSK/8-PSK TCM codes.}
\label{table:bpsk-qpsk-psk8}
\end{table}

%% file: extracted-codes/bpsk-to-qam16-qam64.tex
\begin{table}
\centering
{\scriptsize
\begin{tabular}{|c|c|c||c|c|c|}
\hline
\multicolumn{3}{|c||}{BPSK $\rightarrow$ 16-QAM} & \multicolumn{3}{c|}{BPSK $\rightarrow$ 64-QAM} \\
\hline
$v$  &   $\beta$  &        $G$       &    $v$  &   $\beta$  &        $G$     \\
\hline
3    &    3.42    &        (10 5 4 17)     &    5    &    3.94    &        (2 4 47 40 10 63) \\
4    &    4.15    &        (2 27 10 33)     &    6    &    4.63    &        (100 104 55 40 10 117) \\
5    &    5.05    &        (4 55 10 67)     &    7    &    4.91    &        (210 44 371 60 4 227) \\
6    &    5.31    &        (2 167 100 31)     &    8    &    5.17    &        (362 100 745 52 12 573) \\
7    &    5.8    &        (40 327 20 212)     &    9    &    5.35    &        (10 60 1277 344 16 1747) \\
8    &    6.13    &        (200 757 14 725)     &    10    &    5.61    &        (6 62 3341 3436 12 2473) \\
9    &    6.33    &        (30 1137 346 1453)     &        &        &    \\
10    &    6.63    &        (774 3103 16 2653)     &        &        &    \\
\hline
\end{tabular}
}
\caption{BPSK $\rightarrow$ 16-QAM/64-QAM TCM codes.}
\label{table:bpsk-qam16-qam64}
\end{table}

%% file: extracted-codes/psk8-to-qam16-qam64.tex
\begin{table}
\centering
{\scriptsize
\begin{tabular}{|c|c|c|}
\hline
\multicolumn{3}{|c|}{8-PSK $\rightarrow$ 16-QAM} \\
\hline
$v$  &   $\beta$  &        $G$     \\
\hline
1    &    3.11    &        (0 1 0 1) (0 1 0 0) (1 0 3 1)     \\
2    &    4.36    &        (0 0 0 1) (1 3 0 1) (3 2 3 3)     \\
3    &    5.33    &        (0 1 0 1) (1 3 1 0) (7 1 5 3)     \\
4    &    6.12    &        (0 1 0 1) (3 2 1 4) (4 6 6 5)     \\
5    &    6.12    &        (0 1 0 2) (3 7 3 0) (2 2 4 5)     \\
6    &    6.79    &        (0 1 0 1) (13 13 17 16) (7 6 3 16)     \\
7    &    7.37    &        (0 1 0 1) (10 15 16 2) (17 11 17 23)    \\
8    &    7.37    &        (0 2 0 1) (3 17 5 15) (16 15 20 21)     \\
9    &    7.37    &        (0 1 0 1) (1 5 1 2) (332 10 336 77)     \\
10    &    7.37    &        (15 0 10 16) (15 10 16 13) (5 33 35 15)  \\
\hline\hline
\multicolumn{3}{|c|}{8-PSK $\rightarrow$ 64-QAM} \\
\hline
$v$  &   $\beta$  &        $G$     \\
\hline
3    &    4.15    &        (0 0 0 0 0 1) (0 2 1 3 3 2) (2 3 5 7 1 5) \\
4    &    4.66    &        (0 0 1 0 0 1) (0 1 1 0 0 2) (4 15 0 1 14 0) \\
5    &    5.12    &        (0 0 1 0 0 1) (0 0 1 0 1 3) (4 35 24 13 21 37) \\
6    &    5.53    &        (0 0 1 0 0 1) (0 2 1 0 2 2) (12 53 70 36 75 71) \\
7    &    5.73    &        (0 2 3 0 2 2) (0 2 1 0 2 3) (12 27 10 4 47 6) \\
8    &    5.91    &        (1 0 1 0 0 1) (2 0 0 0 0 3) (172 254 270 6 157 124) \\
9    &    6.26    &        (0 0 1 0 0 1) (0 2 1 0 3 5) (1 272 303 300 21 26) \\
10    &    6.42    &        (0 0 1 0 0 1) (0 2 1 0 6 6) (12 653 670 236 475 471) \\
\hline
\end{tabular}
}
\caption{8-PSK $\rightarrow$ 16-QAM/64-QAM TCM codes.}
\label{table:psk8-qam16-qam64}
\end{table}

%% file: extracted-codes/qam16-qam64.tex
\begin{table}
\centering
{\scriptsize
\begin{tabular}{|c|c|c|}
\hline
\multicolumn{3}{|c|}{16-QAM $\rightarrow$ 64-QAM} \\
\hline
$v$  &   $\beta$  &        $G$     \\
\hline
2    &    3.31    &    (0 0 0 0 0 1) (0 0 1 0 0 1) (2 3 1 1 3 1) (0 3 2 1 0 0) \\
3    &    3.31    &    (0 0 1 0 0 0) (2 3 2 1 0 2) (2 0 1 0 3 0) (2 3 1 1 2 1) \\
4    &    4.18    &    (0 0 1 0 0 0) (0 0 1 0 0 1) (0 3 0 0 2 1) (1 10 11 6 5 0) \\
5    &    4.56    &    (0 0 1 0 0 1) (0 0 1 0 0 0) (0 3 2 1 2 1) (14 24 17 4 3 10) \\
6    &    4.91    &    (0 0 0 0 0 1) (0 0 1 0 0 0) (0 3 1 1 2 0) (32 14 12 30 57 63) \\
7    &    5.23    &    (0 0 1 0 0 1) (0 0 1 0 0 0) (14 17 12 5 16 11) (14 24 17 4 3 10) \\
8    &    5.53    &    (0 0 0 0 0 1) (0 0 1 0 0 1) (26 27 5 15 37 31) (0 23 26 15 0 24) \\
9    &    5.81    &    (0 0 0 0 0 1) (0 0 1 0 0 1) (1 12 2 4 7 10) (55 25 123 42 177 73) \\
10    &    5.81    &    (0 0 0 0 0 1) (0 0 1 0 0 1) (1 6 6 2 7 4) (317 151 641 302 737 36) \\
\hline
\end{tabular}
}
\caption{16-QAM $\rightarrow$ 64-QAM TCM codes.}
\label{table:qam16-qam64}
\end{table}

%% file: extracted-codes/qpsk-to-psk8-qam16-qam64-wide.tex
\begin{table*}[ht]
\centering
{\scriptsize
\begin{tabular}{|c|c|c||c|c|c||c|c|c|}
\hline
\multicolumn{3}{|c||}{QPSK $\rightarrow$ 8-PSK} & \multicolumn{3}{c||}{QPSK $\rightarrow$ 16-QAM} & \multicolumn{3}{c|}{QPSK $\rightarrow$ 64-QAM} \\
\hline
$v$  &   $\beta$  &        $G$     &    $v$  &   $\beta$  &        $G$     &    $v$  &   $\beta$  &        $G$     \\
\hline
1    &    1.12    &    (0 0 1) (1 2 3) 	&	2    &    2.55    &    (1 2 2 1) (1 1 0 3) 	&	4    &    3.01    &    (6 6 1 6 4 12) (1 0 3 2 0 1) \\
2    &    3.01    &    (0 0 1) (2 5 1) 	&	3    &    3.42    &    (0 3 0 1) (6 0 2 5) 	&	5    &    3.41    &    (10 2 15 4 4 12) (4 0 1 2 0 5) \\
3    &    3.6    &    (0 1 2) (2 4 1) 	&	4    &    3.8    &    (0 5 0 2) (7 2 2 7) 	&	6    &    3.68    &    (40 2 72 76 30 75) (0 0 3 2 0 1) \\
4    &    4.13    &    (0 1 3) (6 10 13) &	5    &    4.15    &    (2 10 14 37) (0 3 2 1) &	7    &    3.94    &    (6 26 71 20 44 66) (0 1 3 0 0 5) \\
5    &    4.59    &    (4 17 13) (0 4 1) &	6    &    4.47    &    (1 12 16 11) (1 13 16 7) &	8    &    4.1    &    (0 4 3 6 7 15) (34 32 46 34 54 61) \\
6    &    5.01    &    (0 4 3) (12 1 20) &	7    &    5.05    &    (12 25 5 2) (0 4 0 17) 	&	9    &    4.26    &    (24 16 110 6 62 73) (4 0 7 4 2 12) \\
7    &    5.01    &    (10 75 30) (0 2 7) &	8    &    5.05    &    (15 3 12 25) (12 35 24 10) &	10    &    4.63    &    (12 2 7 4 0 17) (330 102 231 140 46 220) \\
8    &    5.75    &    (42 165 134) (0 6 1) &	9    &    5.56    &    (47 57 36 65) (1 22 2 13) &	&	&	\\
9    &    5.75    &    (0 1 6) (122 250 311) &	10    &    5.56    &    (5 31 30 23) (70 27 10 116) &	&	&	\\
10    &    6.02    &    (0 6 7) (376 227 763) &	&	&	&	&	&	\\
\hline
\end{tabular}
}
\caption{QPSK $\rightarrow$ 8-PSK/16-QAM/64-QAM TCM codes.}
\label{table:qpsk-psk8-qam16-qam64}
\end{table*}

%% file: extracted-codes/searchtime.tex
\begin{table}
\centering
{\scriptsize
\begin{tabular}{|c|c|c|c|}
\hline
$\K$	&	$\N$	&	$T_{\textrm{full}} (s) $	&	$T_{\textrm{random}} (s)$	\\
\hline
BPSK & QPSK   & 1    & $\approx{}0$ \\
BPSK & 8-PSK  & 4    & $\approx{}0$ \\
BPSK & 16-QAM & 263  & 15 \\
QPSK & 8-PSK  & 57   & 44 \\
QPSK & 16-QAM & 7101 & 54 \\

\hline
\end{tabular}
}
\caption{Running time comparison between full search and random search for
codes of constraint length $v=5$.}
\label{table:search-compare}
\end{table}

%% file: ITPSK.tex

%% file: cryptointl.tex
\section{Cryptographic Interleaving}
\label{sec:cryptointl}
Although good TCM codes can improve the system robustness
and hide the modulation, the output symbols
produced by the codes are not indistinguishable to the adversary.
In this section, we propose
``Cryptographic Interleaving'' as a solution for the code concealing problem.
As suggested by the name, we use cryptographic functions for the concealing goal.
Different from conventional data encryption, cryptographic operations are performed on the baseband symbols
(otherwise it is vulnerable to constellation-based modulation-guessing attacks discussed in Section~\ref{sec:approach}).
We devise an efficient method for generating fast cryptographic interleaving functions.

\vskip 1ex\noindent{\bf Interleaving process}
The interleaving process is performed on coded symbols produced by the GTCM module before radio transmission.
The interleaved symbols must
be indistinguishable to the adversary. Specifically, the following requirements
should be met: (1) interleaved symbols look like a sequence produced by a random code,
(2) symbols belonging to different packets are differently permuted,
and (3) the user identity is not revealed. For compactness, we assume the existence
of a shared secret key associated to the undergoing communication session between the transmitter and receiver.
The symbol interleaving map is
randomly selected based on the session's secret key and the transmitted packet number.

For convenience, we assume that the coded symbols produced by the GTCM module
are divided into multiple blocks, each has $m$ symbols, and a fixed number of blocks are packed into a transmitted frame.
To randomly interleave the symbols,
the basic idea is to pre-build a set of interleaving tables
and select one to use based on the tuple information (user key $K$, packet number $s$, block $i$).
The mapping from the tuple $(K,s,i)$ to the interleaving table index can be created
by a hash function (e.g., SHA2).
For example, let $\{I_0,\ldots,I_{N-1}\}$ be the set of $N$ predefined
interleaving tables, and $h$ denote the hash function.
The interleaving table $I_j$ is selected by $j=h(|K|s|i|)\bmod{}N$ and used
to interleave the block $i$ of packet $s$.
As a result, the block $i$'s coded symbols $y_0,\ldots,y_{m-1}$ are permuted into
$y_{I_j(0)},\ldots,y_{I_j(m-1)}$ for transmission.

Although the precomputed tables approach is conceptually simple, it is hard to implement
in practice as one needs to precompute a large number of interleaving functions
so that the adversary cannot guess. This, however, requires significant computation
and storage resources. In the following, we propose a practical and
efficient solution based on a cryptographic linear construction of permutation/interleaving functions.

\vskip 1ex\noindent{\bf Cryptographic linear interleaving}
To eliminate the cost of precomputing interleaving tables, we assume
the number of symbols per block, $m$, is a prime, and define
the interleaving functions as linear functions:
$I_{A,B}(x)=Ax+B\bmod{}m$, where $A\in\{1\twodots{}m-1\},B\in\{0\twodots{}m-1\}$.
It is clear from the construction above that any pair $(A,B)$
will produce a bijective function $I_{A,B}(x)$ with respect to $x$,
i.e., $I_{A,B}(x)$ is an interleaving function.
To randomly generate $A,B$ such that the requirements for indistinguishability are met,
we apply a hash function $h$ on the tuple information $(K,s,i)$ and compute
$$
A=(h(|K|s|i|0)\bmod{}(m-1))+1, \quad{}   B=h(|K|s|i|1)\bmod{}m.
$$
The coded symbols $y_0,\ldots,y_{m-1}$ of block $i$ are
now permuted into $y_{I_{A,B}(0)},\ldots,y_{I_{A,B}(m-1)}$ for transmission.

\vskip 1ex\noindent{\bf Frame format}
Since the interleaving processing on the coded symbols of the user data involves
using not only the secret key $K$, but also the packet number $s$ and block index $i$,
the transmitter needs to embed those information along with the rate information
into the transmitted frame. In the following, we describe the structure of the
physical layer's frame and also discuss the encoding procedure for the frame header.

\begin{center}
\includegraphics[scale=0.32]{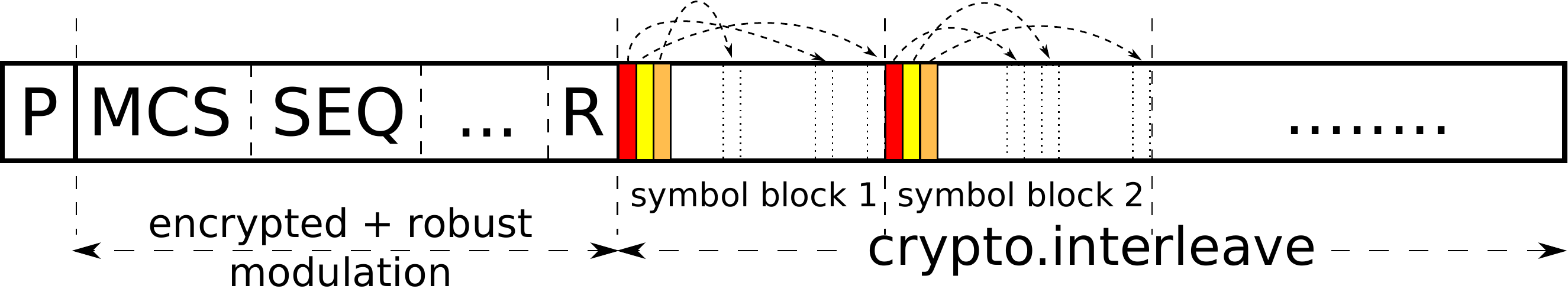}
\end{center}
In the transmitted frame, the preamble $P$ is a publicly known bit sequence (typically 64 bits)
used for frame synchronization at the receiver. The MCS (modulation and coding scheme) field stores
the TCM code's identifier used to encode the data. The SEQ field specifies the packet number
required for the interleaving process. The R field stores a random number generated per packet by the transmitter.
The frame header is encrypted by 
$E_K(\id{MCS}|\id{SEQ}|\ldots|R)$
using AES encryption $E$ with the shared secret key $K$.
The header is encoded by a public robust coding scheme and along with the preamble is modulated
by a public robust modulation. Note that since header and preamble are short, TCM codes are not beneficial
as Viterbi decoder is only applied for long sequences.

\vskip 1ex\noindent{\bf Security and robustness}
Since the interleaving functions are generated based on the cryptographic hash function
with a secret key $K$ applied on varying packet number $s$ and block $i$,
the coefficients $A$ and $B$ are indistinguishable under chosen plaintext attacks (semantically secure),
thus the interleaved symbol sequences are also indistinguishable. As the hash function
is lightweight, the computation of $I_{A,B}$ is extremely fast.
The header is also semantically secure due to the use of random $R$ with AES encryption.
It is also robust as lost synchronization does not propagate to next frames.

%

%% file: evaluation.tex
\section{Evaluation}
\label{sec:evaluation}

\begin{figure}[t]
\centering
\includegraphics[scale=0.4]{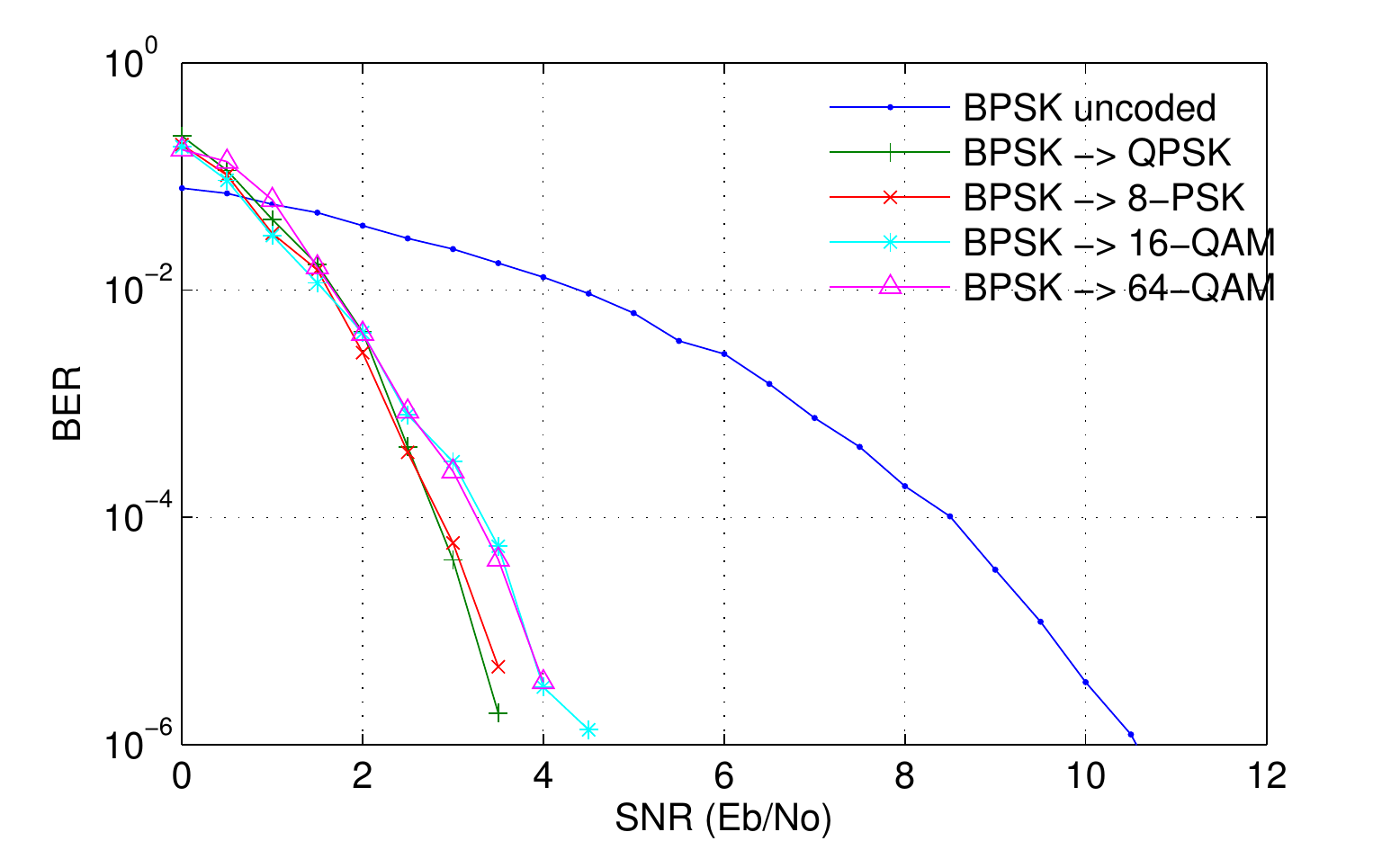}
\caption{Uncoded BPSK vs. coded higher-order modulations (v=10).}
\label{fig:bpsk}
\end{figure}

\begin{figure}[t]
\centering
\includegraphics[scale=0.4]{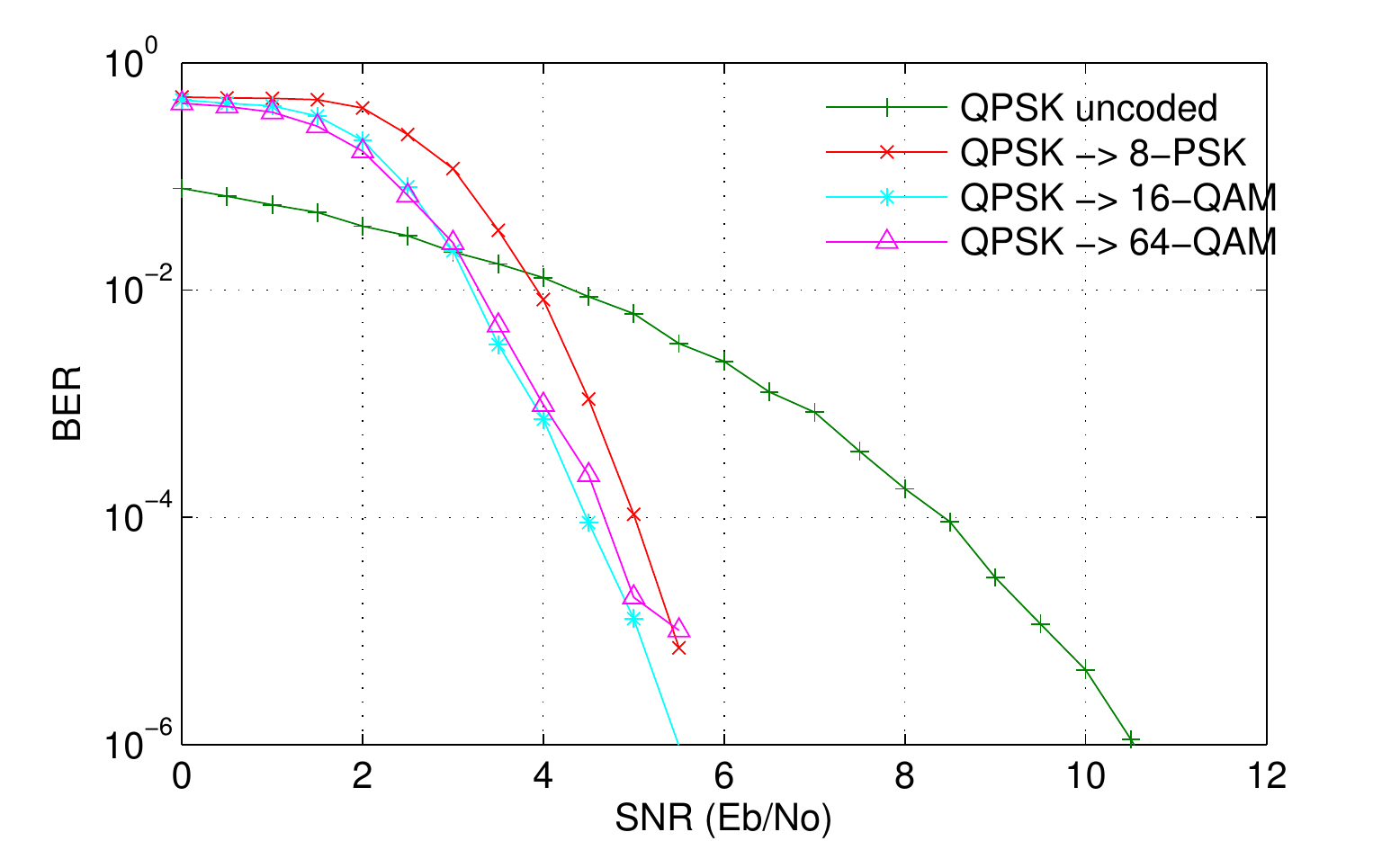}
\caption{Uncoded QPSK vs. coded higher-order modulations (v=10).}
\label{fig:qpsk}
\end{figure}

\begin{figure}[t]
\centering
\includegraphics[scale=0.4]{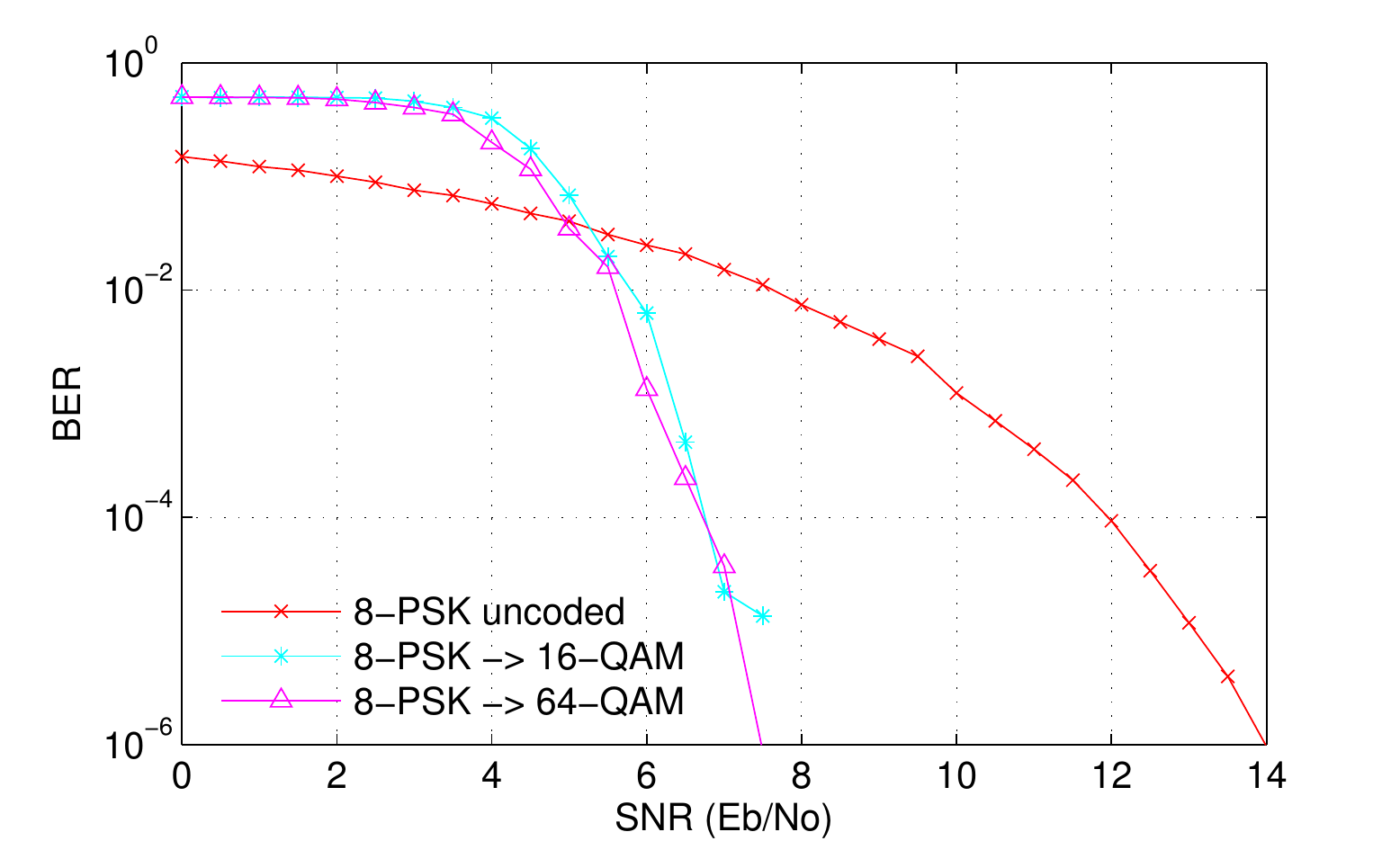}
\caption{Uncoded 8-PSK vs. coded higher-order modulations (v=10).}
\label{fig:psk8}
\end{figure}

\begin{figure}[t]
\centering
\includegraphics[scale=0.4]{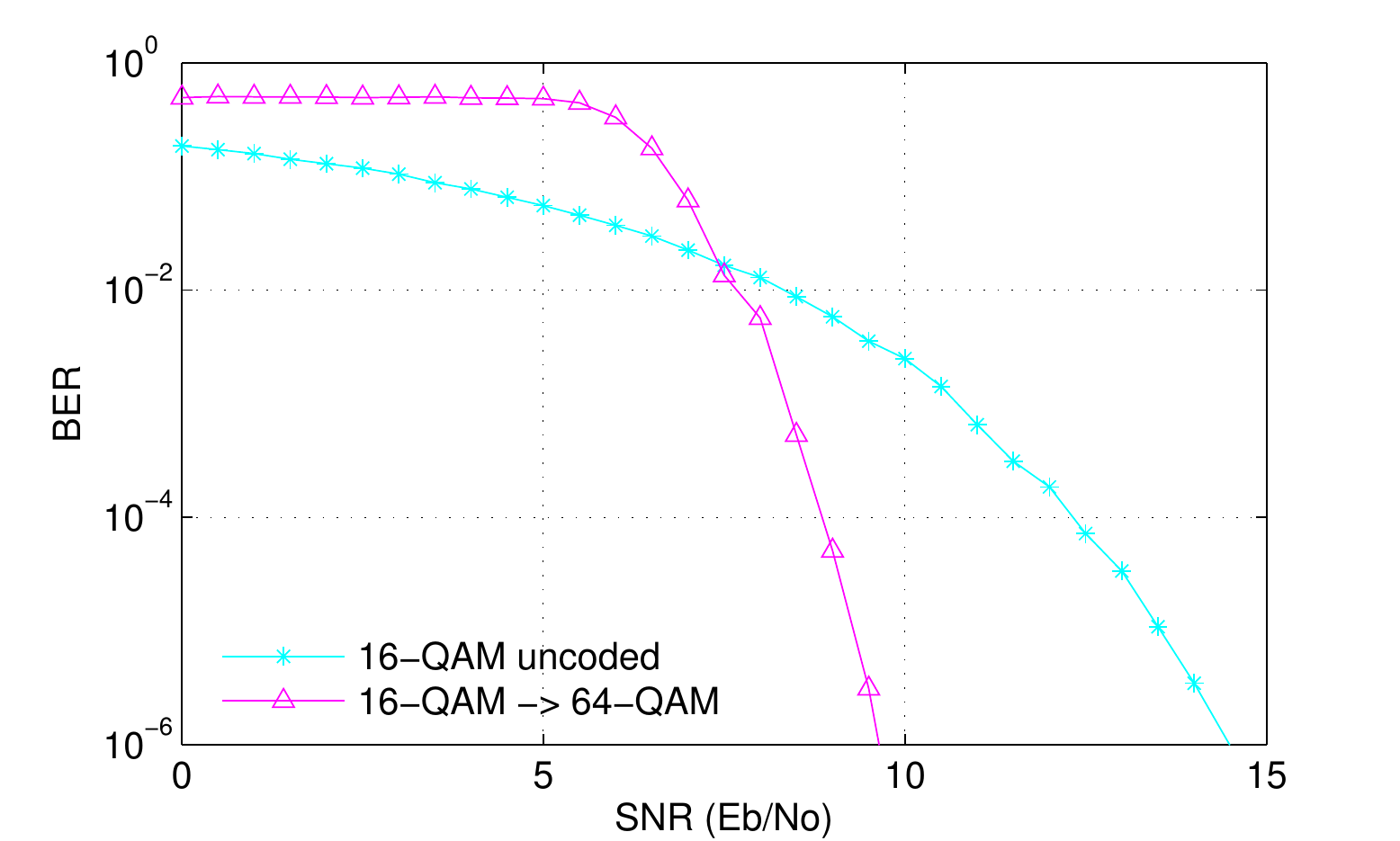}
\caption{Uncoded 16-QAM vs. coded 64-QAM (v=10).}
\label{fig:qam16}
\end{figure}

In this section, we report on the evaluation of the codes found in Section~\ref{sec:gtcm}, in comparison with their corresponding uncoded modulations. The evaluation is performed
through simulations in MatLab. For each pair of the original uncoded modulation $\K$ and
the target coded modulation $\N$, a transmission of 1Gbits is carried in an additive white Gaussian noise
channel. 
We start the simulation at the normalized signal-to-noise ratio $E_b/N_0=0$ and increase it by 0.5dB after every run.
The bit error rates corresponding to each SNR level are recorded.
The simulation stops when all data is done transmitted or 1000 bit errors are reached.

Due to lack of space, we only show the evaluation results for the codes
of constraint length $v=10$. At $\textrm{BER}=10^{-6}$, in addition to hiding the rate,
the coding gain (boosting) provided by the coded modulations ranges from about
$5$dB to more than $6.5$dB when 64-QAM modulation is used for rate concealing. Compared to the modulation unification technique proposed in recent related work~\cite{Rahbari:2014:FCM:2627393.2627415} whose performance degrades by about $1.2$dB
for hiding BPSK modulation in 64-QAM modulation, we gain up to $8$dB.
If the system only supports QPSK as the highest-order modulation, an upgrading BPSK $\rightarrow$ QPSK
can give an advantage of $7.5$dB over uncoded BPSK, while the modulation unification
loses about $2$dB, resulting in our improvement of up to $9.5$dB.
In scenarios where the adversary is weak (i.e., high SNR), 
the coding gain is close to the asymptotic gain about $8.5$dB presented
in Section~\ref{sec:gtcm}, resulting in an improvement of up to $10$dB.

The evaluation results also show that the performance boost is similar
across different target modulations. For example in~\fref{fig:qpsk},
using 8-PSK as the target modulation is within 1dB
of using target modulation 16-QAM or 64-QAM.
This leads to a key lesson that the rate concealing technique based on
coded modulations can be flexibly used in various systems, where different
modulations are supported. One can imagine that in future wireless communication systems always use the highest modulation possible for the RF Front End ADC and AGC, and adapt to the channel conditions only by changing the code.

%% file: conclusion.tex
\section{Conclusion and Discussion}

We proposed a solution to the problem of hiding the rate of a communication while simultaneously increasing the robustness of the communication to interference. Our approach relies on algorithms for discovering new General TCM codes, and a cryptographic interleaving scheme. These algorithms include new efficient techniques to determine the free distance of non-uniform TCM codes. We explicitly derived 85 codes for upgrading any modulation in \{BPSK, QPSK, 8-PSK, 16-QAM, 64-QAM\} into any higher order modulation. These are the best codes among uniform and non-uniform TCM codes specifically designed for coded modulation and that conceal the underlying rate of the communication. We demonstrate that beyond achieving rate-hiding, these codes also provide an order of magnitude improvement of energy efficiency in comparison with recent related work.
The proposed codes and cryptographic interleaving techniques are easily deployable in software defined radios since they only necessitate a small table per code. Beyond defending against rate adaptation attacks and boosting the performance of wireless systems at a time where RF spectrum is scarce, they also mitigate against passive attacks against users traffic analysis~\cite{AtkinsonARMM13}.